\documentclass{article}
\usepackage{graphicx}
\usepackage{amsthm,amsmath,amssymb} 
\usepackage{color}
\usepackage[mathscr]{euscript}
\usepackage[noadjust]{cite}
\usepackage{multirow}
\usepackage{tabu,array,soul}
\usepackage[a4paper, total={5.5in, 7.5in}]{geometry}
\usepackage{subcaption}
\usepackage{caption}
\usepackage{algorithm}
\usepackage{algpseudocode}
\usepackage{siunitx}
\usepackage{textcomp}
\usepackage[table]{xcolor}

\usepackage{soul}

\newtheorem{theorem}{Theorem}
\newtheorem{result}{Result}
\newtheorem{observation}{Observation}
\newtheorem{definition}{Definition}
\renewcommand*{\theobservation}{\Alph{observation}}

\newtheorem{corollary}{Corollary}

\newcounter{claimcounter}

\definecolor{Gray}{gray}{0.5}
\newcolumntype{a}{>{\columncolor{Gray}}c}

\newenvironment{claim}{\stepcounter{claimcounter} \medskip\noindent\textit{Claim \theclaimcounter. }\itshape}

\newtheorem{lemma}{Lemma}

\newcommand{\com}[1]{}

\title{Approximating Minimum Dominating Set on String Graphs}

\date{}
\author{Dibyayan Chakraborty\thanks{Indian Statistical Institute, Kolkata, India. E-mail: \texttt{dibyayancg@gmail.com}.}\and Sandip Das \thanks{Indian Statistical Institute, Kolkata, India. E-mail: \texttt{sandipdas69@isical.ac.in}.}\and Joydeep Mukherjee\thanks{Indian Statistical Institute, Kolkata, India. E-mail: \texttt{joydeep.m1981@gmail.com}. Partially supported by DST SERB NPDF fellowship (PDF/2016/001647).} }

\begin{document}

\maketitle

\begin{abstract}
	In this paper, we give approximation algorithms for the \textsc{Minimum Dominating Set (MDS)} problem on \emph{string} graphs and its subclasses. A \emph{path} is a simple curve made up of alternating horizontal and vertical line segments. A \emph{$k$-bend path} is a path made up of at most $k + 1$ line segments. An \textsc{L}-path is a $1$-bend path having the shape `\textsc{L}'. A \emph{vertically-stabbed-\textsc{L} graph} is an intersection graph of \textsc{L}-paths intersecting a common vertical line. We give a polynomial time $8$-approximation algorithm for \textsc{MDS} problem on vertically-stabbed-\textsc{L} graphs whose APX-hardness was shown by Bandyapadhyay et al. (\textsc{MFCS}, 2018). To prove the above result, we needed to study the \emph{Stabbing segments with rays} (\textsc{SSR}) problem introduced by Katz et al. (\textsc{Comput. Geom. 2005}). In the \textsc{SSR} problem, the input is a set of (disjoint) leftward-directed rays, and a set of (disjoint) vertical segments. The objective is to select a minimum number of rays that intersect all vertical segments. We give a $O((n+m)\log (n+m))$-time $2$-approximation algorithm for the \textsc{SSR} problem where $n$ and $m$ are the number of rays and segments in the input. A \emph{unit $k$-bend path} is a $k$-bend path whose segments are of unit length. A graph is a \emph{unit $B_k$-VPG graph} if it is an intersection graph of unit $k$-bend paths. Any string graph is a unit-$B_k$-VPG graph for some finite $k$. Using our result on \textsc{SSR}-problem, we give a polynomial time $O(k^4)$-approximation algorithm for \textsc{MDS} problem on unit $B_k$-VPG graphs for $k\geq 0$.
\end{abstract}


\section{Introduction}

A graph $G$ has vertex set $V(G)$ and edge set $E(G)$. A \emph{dominating set} of a graph $G$ is a subset $D\subseteq V(G)$ of vertices such that each vertex in $V(G)\setminus D$ is adjacent to some vertex in $D$. The \textsc{Minimum Dominating Set (MDS)} problem is to find a minimum cardinality dominating set of a graph $G$. A \emph{string representation} of a graph is a collection of simple curves on the plane such that each curve in the collection represents a vertex of the graph and two curves intersect if and only if the vertices they represent are adjacent in the graph. The graphs that have a string representation are called \emph{string} graphs. 

Many important graph families like \emph{planar} graphs, \emph{unit disk} graphs and \emph{chordal graphs} are subclasses of string graphs~\cite{asinowski2012,goncalves2017}. Indeed, Pach and Toth~\cite{pach2003} proved that the number of string graphs on $n$ labelled vertices is at least $2^{\frac{3}{4}{\binom{n}{2}}}$, indicating that many graphs are string graphs. 
This motivates the search for efficient algorithms for solving optimisation problems on string graphs. Fox and Pach~\cite{fox2011} gave for every $\epsilon > 0$, a polynomial time algorithm for computing the \textsc{Maximum Independent set} of $k$-string graphs (intersection graphs of curves on the plane where two curves intersecting at most $k$ times) with approximation ratio at most $n^\epsilon$. While Pawlik et al.~\cite{pawlik2014} proved that triangle-free \emph{segment graphs} (intersection graphs of line segments on the plane) can have arbitrarily high \textsc{Chromatic Number}, Bonnet et al.~\cite{BonnetR18} gave a subexponential algorithm to color string graphs with three colors. In this paper, we study the \textsc{MDS} problem on string graphs and its subclasses.

Since split graphs (graphs whose vertex set can be partitioned into a clique and an independent set) are known to be string graphs, for every $\alpha>0$, it is not possible to approximate the \textsc{MDS} problem on string graphs with $n$ vertices to within $(1-\alpha)\ln n$ unless $NP\subseteq DTIME(n^{O(\log \log n)})$~\cite{chlebik2008}. Hence, researchers have focussed on developing approximation algorithms for the \textsc{MDS} problem on special classes of string graphs. The concepts of \emph{bend number} and \emph{$B_k$-VPG graphs} (introduced by Asinowski et al.~\cite{asinowski2012}) become useful in gaining a better understanding of subclasses of string graphs. A \emph{path} is a simple curve made up of alternating horizontal and vertical line segments. A \emph{$k$-bend path} is a path made up of at most $k + 1$ line segments. A \emph{$B_k$-VPG representation} of a graph is a collection of $k$-bend paths such that each path in the collection represents a vertex of the graph, and two such paths intersect if and only if the vertices they represent are adjacent in the graph. The graphs that have a $B_k$-VPG representation are called \emph{$B_k$-VPG graphs}. A graph is said to be a \emph{VPG graph} if it is a $B_k$-VPG graph for some $k$. Asinowski et al.~\cite{asinowski2012} showed that the family of VPG graphs are equivalent to the family of string graphs. 

Mehrabi~\cite{mehrabi2017} gave an $\epsilon$-net based $O(1)$-approximation algorithm for the MDS problem on \emph{one-string} $B_1$-VPG graphs (graphs with $B_1$-VPG representation where two curves intersect at most once). Bandyapadhyay et al.~\cite{Bandyapadhyay2018} proved APX-hardness for the MDS problem on a special class of $B_1$-VPG graphs, namely \emph{vertically-stabbed-\textsc{L} graph} (defined below) which was originally introduced by McGuinness~\cite{mcguinness1996}. An \textsc{L}-path is a $1$-bend path having the shape `\textsc{L}'. A \emph{vertically-stabbed-\textsc{L}-representation} of a graph is a collection of \textsc{L}-paths and a vertical line such that each path in the collection intersects the vertical line. Each path in the collection represents a vertex of the graph and two paths intersect if and only if the vertices they represent are adjacent in the graph. A graph is a \emph{vertically-stabbed-\textsc{L} graph} if it has a vertically-stabbed-\textsc{L}-representation. Bandyapadhyay et al.~\cite{Bandyapadhyay2018} proved APX-hardness for the MDS problem on vertically-stabbed-\textsc{L} graphs by showing that all \emph{circle} graphs (intersection graphs of chords of a circle) are vertically-stabbed-\textsc{L} graphs. Many researchers have studied the \textsc{MDS} problem on circle graphs~\cite{bousquet2012,damian2006,damian2002}. Since all vertically-stabbed-\textsc{L}-graphs are also one-string $B_1$-VPG graphs, there is a $O(1)$-approximation algorithm for the MDS problem on vertically-stabbed-\textsc{L} graphs (due to Mehrabi~\cite{mehrabi2017}). In this paper, we prove the following theorems.

\begin{theorem}\label{thm:stab-L-dom}
	Given a vertically-stabbed-\textsc{L}-representation of a graph $G$ with $n$ vertices, there is a polynomial time $8$-approximation algorithm to solve the \textsc{MDS} problem on $G$.
\end{theorem}

The time complexity of our algorithm for the \textsc{MDS} problem on  vertically-stabbed-\textsc{L} graph is essentially the time required to solve a linear program with $(0,1)$-coefficient matrix optimally. To prove the above theorem we needed to study the \emph{stabbing segment with rays (\textsc{SSR})} problem introduced by Katz et al.~\cite{katz2005}. In the \textsc{SSR} problem, the inputs consist of a set of (disjoint) leftward-directed rays and a set of (disjoint) vertical segments. The objective is to select a minimum number of leftward-directed rays that intersect all vertical segments. Katz et al.~\cite{katz2005} gave a dynamic programming based $O(n^2(n+m))$-time optimal algorithm to solve the \textsc{SSR} problem where $n$ and $m$ are the number of rays and segments in the input instance respectively. We prove the following theorem.

\begin{theorem}\label{thm:SSR}
	There is a $O((n+m)\log (n+m))$-time $2$-approximation algorithm for \textsc{SSR} problem where $n$ and $m$ are the number of rays and segments in the input instance respectively.
\end{theorem}

Using our approximation algorithm for the \textsc{SSR} problem, we give approximation algorithms for other subclasses of string graphs as well. In this paper, we introduce the class \emph{Unit-$B_k$-graph} as follows. A \emph{unit $k$-bend path} is a $k$-bend path whose segments are of unit length. A \emph{unit $B_k$-VPG representation} of a graph $G$ is a $B_k$-VPG representation $\mathcal{R}$ of $G$ where all paths in $\mathcal{R}$ are unit $k$-bend paths. Graphs having unit $B_k$-VPG representations are called \emph{unit $B_k$-VPG graphs}. A graph is said to be a \emph{UVPG graph} if it is a unit $B_k$-VPG graph for some $k$. Notice that, every $B_k$-VPG graph has a unit $B_{k'}$-VPG representation for some finite $k'\geq k$. The family of UVPG graphs are equivalent to the family of VPG graphs and therefore equivalent to the family of string graphs. Using our approximation algorithm for the \textsc{SSR} problem, we prove the following theorem.

\begin{theorem}\label{thm:unit-bk-VPG}
	Given a unit $B_k$-VPG representation of a graph $G$ with $n$ vertices, there is a polynomial time $O(k^4)$-approximation algorithm to solve the \textsc{MDS} problem on $G$. 
\end{theorem}
On the negative side, we shall show that solving \textsc{MDS} problem on Unit-$B_1$-VPG graph is NP-Hard. First, we give the approximation algorithm for the \textsc{SSR} problem and prove Theorem~\ref{thm:SSR} in Section~\ref{sec:SSR}. In Section~\ref{sec:vertical-L}, we prove Theorem~\ref{thm:stab-L-dom} and in Section~\ref{sec:unit-VPG}, we give a proof sketch for Theorem~\ref{thm:unit-bk-VPG}. Finally we draw conclusion in Section~\ref{sec:conclude}.

\section{Approximation for SSR-problem}\label{sec:SSR}

Throughout this section, we let $\mathcal{SSR}(R,V)$ denote an \textsc{SSR} instance where $R$ is a given set of (disjoint) leftward-directed rays and $V$ is a given set of (disjoint) vertical segments. The objective is to select a minimum cardinality subset of $R$ that intersects all segments in $V$. In this section, unless otherwise stated, whenever we say a ``ray'' we shall refer to a leftward-directed ray and whenever we say a ``segment'' we shall refer to a vertical segment. Without loss of generality, we can assume that all segments lie in the first quadrant of the plane, each segment intersects at least one ray and no two segments in $V$ has same $x$-coordinate. 

\medskip

\noindent\textbf{The algorithm:} Our algorithm consists of four main steps. (a) If some segments in $V$ intersect precisely one ray $r\in R$, we put $r$ in our heuristic solution $S$. (b) We delete all segments intersecting any ray in $S$ from $V$. (c) We find a ray in $R\setminus S$ whose $x$-coordinate of the right endpoint is the smallest among all rays in $R\setminus S$ and delete it from $R$ (when there are multiple such rays, choose anyone arbitrarily). We repeat steps (a)-(c) until all segments are deleted from the instance. We shall refer to the above algorithm as \emph{\textsc{SSR}-Algorithm}. Notice that, the \textsc{SSR}-Algorithm takes $O((n+m)\log (n+m))$ time (using segment tree data structure~\cite{berg2008}) where $n$ and $m$ are the number of rays and segments in the input. 

\medskip

We shall show that the cardinality of the set returned by the \textsc{SSR}-algorithm is at most twice the optimum cost of the relaxed linear programming of the input instance. This would give us the desired  approximation factor of the \textsc{SSR}-algorithm. First, we introduce some definitions and some extra steps in the above algorithm for analysis purpose. With each ray $r\in R$ we associate a \emph{token} $T_r$, which is a subset of $R$ (possibly empty). We add an \emph{Initialisation Step} to our algorithm where we assign $T_r=\{r\}$ for each $r\in R$ and assign $R_0$ as $R$, $V_0$ as $V$, $S_0$ as $\emptyset$. For $i\geq 1$ let $R_i,V_i,S_i$ be the set of rays, the set of segments and the heuristic solution constructed by the SSR-Algorithm, respectively at the end of $i^{th}$ iteration. A ray $r\in R_i$ is \emph{critical} if there is a segment $v\in V_i$ such that $r$ is the only ray in $R_i$ that intersects $v$. Notice that, in the $i^{th}$ iteration step (a) of the SSR-Algorithm collects all the rays that became critical at the end of $(i-1)^{th}$ iteration and adds them to the heuristic solution to create $S_{i}$. Let $D$ be a subset of $R_i$. A ray $r\in D$ lies \emph{in between} two rays $r',r''\in D$ if the $y$-coordinate of $r$ lies in between those of $r',r''$. A ray $r\in D$ lies \emph{just above} (resp. \emph{just below}) a ray $r'\in D$ if $y$-coordinate of $r$ is greater (resp. smaller) than that of $r'$ and no other ray lies in between $r,r'$ in $ D $. Two rays $r,r'\in D$ are \emph{neighbours} of each other if $r$ lies just above or below $r'$.

\begin{definition}
	For a ray $r\in R_{i-1}\setminus S_i$ and $i\geq 1$, the phrase ``$r$ passes the token to its neighbours" in the $i^{th}$ iteration shall refer to the following operations in the prescribed order:
	\begin{enumerate}
		\renewcommand{\theenumi}{(\roman{enumi})}
		\renewcommand{\labelenumi}{\theenumi}
		
		\item Let $r'$ lies just above $r$ and $r''$ lies just below $r$ in $R_{i-1}\setminus S_i$. For all $x\in T_r$ ($x$ and $r$ not necessarily distinct) do the following. If there is a segment in $V_{i}$ that intersects $x,r'$ and $r$ then assign $T_{r'}=T_{r'}\cup \{x\}$ and if there is a segment in $V_{i}$ that intersects $x,r''$ and $r$ then $T_{r''}=T_{r''}\cup \{x\}$. 
		
		\item $T_r=\emptyset$.
	\end{enumerate}
	
\end{definition}

\begin{algorithm}
	\caption{MOD-SSR-Algorithm}\label{alg:SSR}
	\hspace*{\algorithmicindent} \textbf{Input:} A set $R$ of leftward-directed rays and a set $V$ of vertical segments. \\
	\hspace*{\algorithmicindent} \textbf{Output:} A subset of $R$ that intersects all segments in $V$. 
	\begin{algorithmic}[1]
		\State $T_r=\{r\}$ for each $r\in R$ and $i\leftarrow 1,V_0\leftarrow V,R_0\leftarrow R, S \leftarrow\emptyset,S_0\leftarrow \emptyset$ \Comment{Initialisation.}
		\While{ $V_{i-1}\neq \emptyset$}
		\State $S\leftarrow S\cup\{r\colon r\in R_{i-1}, r~\text{ is critical after}~(i-1)^{th}\text{ iteration}\}$. 
		\medskip
		
		\Comment{Critical ray collection (step (a) of SSR-Algoritm).}
		\medskip
		
		\State $S_{i}\leftarrow S$. \Comment{$S_i$ is $S$ at the $i^{th}$ iteration.}
		\State $V_{i}\leftarrow$ the set obtained by deleting all segments from $V_{i-1}$ that intersect some rays in $S_{i}$. \medskip

		\Comment{Step (b) of SSR-Algorithm.} \medskip
		
		%        \State Shorten every ray in $R_{i-1}\setminus S_{i}$ such that the $x$-coordinates of the right endpoint lies on a segment or on the $y$-axis. 
		%        \medskip 
		
		%        \Comment{Step (c) of SSR-Algorithm.}
		%        \medskip
		
		\State Find a $r\in R_{i-1}\setminus S_{i}$ whose $x$-coordinate of the right endpoint is the smallest.
		
		\State $r$ passes the token to its neighbours. 
		
		\State $R_{i}\leftarrow $ The set obtained by deleting $\{r\}\cup S_{i}$ from $R_{i-1}$. 
		\medskip
		
		%
		%        \Comment{Token-passing step (modified step (d) of SSR-Algorithm). }
		%        \medskip
		
		\Comment{Token-passing step (modified step (c) of SSR-Algorithm). }
		\medskip
		
		\State $i\leftarrow i+1$;
		\EndWhile
		\State \textbf{return} $S$
	\end{algorithmic}
\end{algorithm}

\begin{figure}[ht]
	\centering
	\scalebox{0.85}{
	\begin{tabular}{cccc}
		\includegraphics[page=3]{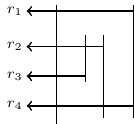}&\includegraphics[page=4]{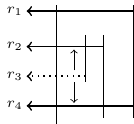}&\includegraphics[page=5]{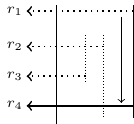}&\includegraphics[page=6]{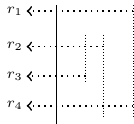}\\
		(a) & (b) & (c) & (d)
	\end{tabular}}
	\caption{(a) An input \textsc{SSR} instance, (b) $1^{st}$ iteration, (c) $2^{nd}$ iteration and (d) $3^{rd}$ iteration of the \textsc{MOD-SSR}-Algorithm with (a) as input. A dotted ray (or segment) indicates a deleted ray (or segment). }\label{fig:example-SSR-algo}
\end{figure}

The modified algorithm is stated in Algorithm~\ref{alg:SSR}. For an illustration, consider the input instance shown in Figure~\ref{fig:example-SSR-algo}(a). At the first iteration, $x$-coordinate of the right endpoint of $r_3$ is the smallest. So, $r_3$ passes the token to its neighbours ($r_2,r_4$) and gets deleted. At the end of $1^{st}$ iteration, notice that $r_2$ has become critical. At the beginning of the $2^{nd}$ iteration Algorithm~\ref{alg:SSR} put $r_2$ in the heuristic solution. Then all segment intersecting $r_2$ is deleted and $r_2$ itself is also deleted. Also in the second iteration, $r_1$ passes the token to its neighbour ($r_4$) and gets deleted. Finally in the third iteration $r_4$ is put in the heuristic solution. 

It is easy to see that, the rays returned by Algorithm~\ref{alg:SSR} indeed intersect all segments of the input SSR instance. Moreover for each $i$, the sets $R_i,S_i,V_i$ constructed by Algorithm~\ref{alg:SSR} is same as the sets constructed by \textsc{SSR}-algorithm. From now on, whenever we refer to the sets $R_i,S_i,V_i$ we shall refer to the sets constructed by Algorithm~\ref{alg:SSR}. We shall now prove that the cardinality of the solution returned by Algorithm~\ref{alg:SSR} is at most twice the optimum solution size of the input SSR instance. Below we describe some observations.

\begin{observation}\label{obs:token-pass}
	For any ray $r\in R$ and some integer $k\geq 1$, if there is a segment $v\in V_k$ that intersects $r$ in the input instance, then $v$ also intersects some ray $r'\in R_k$ such that $r\in T_{r'}$. 
\end{observation}
\begin{proof}
	Fix an arbitrary segment $v\in V_k$ that intersects $r$ in the input instance. Let $Y$ be the set of all rays that intersects $v$ in the input instance and $X\subseteq Y$ be the set of rays that passed the token to its neighbours at some $i^{th}$ iteration with $i<k$. Formally $X=\{r'\in Y\colon T_{r'}=\emptyset\}$. If $r\notin X$, then $r\in R_k$ and we are done. Otherwise, $r\in X$. Let $<r_1,r_2,...,r_t>$ be a sorted order of the rays in $X$ such that for $i<j$, $r_i$ passed the token to the neighbours before $r_j$. Due to step $6$ of \textsc{MOD-SSR}-algorithm, for $i<j$, the $x$-coordinate of the ray $r_i$ is less than that of $r_j$. Hence $<r_1,r_2,...,r_t>$ is an increasing sequence based on the $x$-coordinate of their right endpoint. This implies that whenever a ray $r_i\in Y$ was deleted it was \emph{shortest} among the rays that intersect $v$ and belonged to the set $R_{i-1}$. Therefore, whenever a ray $r_i\in X$ passed the token to its neighbours, it had a neighbor $r'_i\in R_i$ which intersects $v$. This fact combined with the token passing rule in Definition $1$ says that $v$ intersects a ray $r'\in R_{i}$ such that $r\in T_{r'_i}$. Applying the above arguments for all rays in $X$, we have the proof.    
\end{proof}

For any integer $k\geq 0$ and a ray $r\in R$, let $A_k(r)$ denote the segments in $V_{k}$ that intersect $r$ and let $B_k(r)$ be the rays in $R_{k}$ that intersect at least one segment in $A_k(r)$.

%For any integer $k\geq 1$ and a ray $r\in R$, let $A_k(r)$ denote the set of segments in $V_{k}$ that intersect $r$ in the input instance and let $B_k(r)$ be the set of rays in $R_{k-1}\setminus S_k$ that intersect some segments in $A_k(r)$.

\begin{observation}\label{obs:bottom-top}
	For any ray $r\in R$ and some integer $k\geq 1$, let $c$ be the bottom-most (resp. top-most) ray in $B_{k-1}(r)$. If (i) $c$ is the only ray in $B_{k-1}(r)\setminus S_k$ whose token contains $r$ and (ii) $c$ passes the token to its neighbours in the $k^{th}$ iteration, then (a) there is exactly one ray $c'\in R_k$ whose token contains $r$; and $c'$ is the bottom-most (resp. top-most) ray in $B_{k}(r)$.
\end{observation}
\begin{proof}
	Since $c$ is the bottom-most (resp. top-most) ray in $B_{k-1}(r)$ and $c\in B_{k-1}(r)\setminus S_k$, then $c$ must be bottom-most (resp. top-most) ray in $B_{k-1}(r)\setminus S_k$. There is exactly one neighbour $c'$ of $c$ in $R_{k-1}\setminus S_k$ such that there is a segment in $V_k$ that intersects $r,c,c'$. Therefore, when $c$ passes the token to its neighbours, $c'$ is the only ray whose token shall contain $r$. Since $c$ does not belong to $R_k$, $c$ does not belong to $B_{k}(r)$. Hence, $c'$ must be the bottom-most (resp. top-most) ray in $B_{k}(r)$. This completes the proof of the claim.
\end{proof}
\begin{lemma}\label{lem:at-most-two-copy}
	Let $r$ be a ray. After the termination of Algorithm~\ref{alg:SSR}, there are at most two tokens containing $r$.
\end{lemma}

\begin{proof}
	If $r$ never passed the token to its neighbours, then the only token that contains $r$ is $T_r$ and therefore the statement is true. Now let us assume that, at iteration $i$, $r$ passed the token to its neighbours and was deleted from $R_{i-1}$. Let $j$ be the minimum integer with $i<j$ such that at the end of $(j-1)^{th}$ iteration, there is a ray $p\in R_{j-1}$ which is critical and $r\in T_p$. If $j$ does not exist then we simply assume $j$ to be the number of iterations performed by \textsc{MOD-SSR}-Algorithm. We shall prove the following claims.
	
	\begin{claim}
		For each $k<i$, the only ray in $R_{k}$ whose token contains $r$ is $r$ itself.
	\end{claim}
	
	\medskip
	\noindent The proof of the above claim follows directly from the initialisation step and Token-passing step of the Algorithm~\ref{alg:SSR}.
	
	\begin{claim}
		For each $k$ with $i\leq k<j$, (a) there at most two rays $r', r''\in R_{k}$ such that $r\in T_{r'}\cap T_{r''}$; (b) if both $r',r''$ exists then they are neighbours and $r$ lies in between $r',r''$ in $R_0$; (c) if there is exactly one ray $r'''\in R_{k}$ such that $r\in T_{r'''}$ then $r'''$ must be the top-most or bottom-most ray in $B_{k}(r)$.
	\end{claim}  
	
	\medskip
	\noindent We prove the claim by induction on the number of iterations. Consider the $i^{th}$ iteration, where $r$ passed the token to its neighbours and got deleted from $R_{i-1}$ to create $R_{i}$. Clearly, if $r$ was not the top-most or bottom-most ray of $B_{i-1}(r)$, then at the end of $i^{th}$ iteration there are at most two rays $r_1,r_2\in R_{i}$ such that $r_1,r_2$ are neighbours, $r\in T_{r_1}\cap T_{r_2}$ and $r$ lies in between $r_1,r_2$ in $R_0$. Hence, the claim remains true in this case. Suppose $r$ was the top-most (resp. bottom-most) ray of $B_{i-1}(r)$. Since $r\notin S_i$, $r$ is the top-most (resp. bottom-most) ray of $B_{i-1}(r)\setminus S_i$. By Observation~\ref{obs:bottom-top}, at the end of $i^{th}$ iteration there is exactly one ray $r_3\in R_{i}$ such that $r\in T_{r_3}$ and $r_3$ must be the top-most (resp. bottom-most) ray in $B_{i}(r)$. Hence, the claim remains true in this case. We assume the claim to be true for all $i,i+1,\ldots,(k-1)^{th}$ iterations. Let $x$ passed the token to its neighbours in the $k^{th}$ iteration. If $r\notin T_x$ then the claim remains true. When $r\in T_x$, we have the following cases. 
	\begin{enumerate}
		\renewcommand{\theenumi}{(\roman{enumi})}
		\renewcommand{\labelenumi}{\theenumi}
		
		\item Let $x$ be the only ray in $R_{k-1}$ such that $r\in T_x$. Then by induction hypotheis, $x$ was the top-most (or bottom-most) ray in $B_{k-1}(r)$ and hence in $B_{k-1}(r)\setminus S_k$. By Observation~\ref{obs:bottom-top}, at the end of $k^{th}$ iteration there is exactly one ray $x'\in R_{k}$ such that $r\in T_{x'}$ and $x'$ must be the top-most (resp. bottom-most) ray in $B_{k}(r)$. 
%Letting $r'''=x'$, we have the proof of the claim.
		
		\item Let $x_1,x_2\in R_{k-1}$ be two rays such that $r\in T_{x_1}\cap T_{x_2}$. Without loss of generality, we further assume that $x=x_1$, $x_1$ lies just above $x_2$. If there exists a neighbour of $x_1$ (say $x_3$) which is different from $x_2$, then due to the Token-passing step of $k^{th}$ iteration, $x_1$ passes the token to its neighbours (i.e $x_2$ and $x_3$) and gets deleted from $R_{k-1}$ to create $R_{k}$. 
%Hence $x_2,x_3\in R_{k}$, $r\in T_{x_2}\cap T_{x_3}$, $x_2,x_3$ are neighbours and $r$ lies in between $x_2,x_3$ in $R_0$. 
		Letting $r'=x_2$ and $r''=x_3$, we have the proof of the claim. If $x_3$ does not exist, then $x_1$ shall pass the token only to $x_2$ and $x_2$ becomes the top-most ray in $R_k$ (and therefore in $B_k(r)$).
	\end{enumerate}
	This completes the proof of the claim.    
	
	\begin{claim}
		For each $k\geq j$, there is at most one ray $r'\in R_{k}$ such that $r\in T_{r'}$ and if $r'$ exists then $r'$ must be the top-most or bottom-most ray in $B_k(r)$.
	\end{claim}
	
	\medskip \noindent We shall prove the claim by induction on the number of iteration. First consider the $j^{th}$ iteration. Due to Claim~2, we know that there was at most two rays $r',r''\in R_{j-1}$ such that $r\in T_{r'}\cap T_{r''}$.  Recall from the definition of $j$ that, there was a ray $p\in R_{j-1}$ which was critical and $r\in T_p$. We have the following cases.
	\begin{enumerate}
		\renewcommand{\theenumi}{(\roman{enumi})}
		\renewcommand{\labelenumi}{\theenumi}
		
		\item If there were only one ray $r'''\in R_{j-1}$ whose token contained $r$, then $p$ and $r'''$ must be same. Since $p$ is put in heuristic solution, $p$ will never pass the token to its neighbours at any subsequent iteration. Therefore only $T_p$ will contain $r$ after the termination of Algorithm~\ref{alg:SSR}.
		
		\item Let both $r',r''\in R_{j-1}$ exists. By Claim~2 they must be neighbours and $r$ lies in between $r',r''$ in $R_0$. Without loss of generality, assume that $r'$ lies just above $r''$ and $p=r''$. If both $r',r''\in S_j$, then there is nothing to prove. Otherwise, $r'$ is the only ray in $R_{j-1}\setminus S_j$ whose token contains $r$. By Observation~\ref{obs:token-pass}, any segment of $V_j$ that intersects $r$ in the input instance, intersects $r'$. Now consider the set $A_j(r)$ and the set $B_{j-1}(r)\setminus S_j$ and notice that $r'$ must be the bottom-most ray in $B_{j-1}(r)\setminus S_j$. If $r'$ did not pass the token to its neighbours in the $j^{th}$ iteration, then $r'$ becomes the bottom-most ray in $B_j(r)$ and the statement of the claim remains true. If $r'$ passed the token to its neighbours in the $j^{th}$ iteration, then we are done by Observation~\ref{obs:bottom-top}.
	\end{enumerate}

	Now assume that the statement of the claim remains true for $j,j+1,\ldots,(k-1)^{th}$ iteration. If there is no ray in $R_{k-1}\setminus S_k$ whose token contained $r$ then we directly have the proof of the claim. Otherwise by induction hypothesis, there is a unique ray $r'\in R_{k-1}$ whose token contains $r$ and if $r'$ exists then $r'$ is the bottom-most or top-most ray in $B_{k-1}(r)$. Again if $r'$ did not pass the token to its neighbours in the $k^{th}$ iteration, then the statement of the claim remains true. If $r'$ passed the token to its neighbours in the $k^{th}$ iteration, then $r'$ must be the bottom-most or top-most ray in $B_{k-1}(r)\setminus S_k$ and by Observation~\ref{obs:bottom-top} we have the proof of the claim.
	
	Now combining Claim~1, Claim~2 and Claim~3, we have the proof of the lemma. 
\end{proof}
For a segment $v\in V$, we let $N(v)$ denote the set of rays in $R$ that intersect $v$. Let $r\in S$ be a ray, $i$ be the minimum integer such that $r\in S_i$. In other words, $r$ was put in the heuristic solution in the $i^{th}$ iteration. This means there must exist a segment $\nu_r\in V_{i-1}$ such that $r$ is the only ray in $R_{i-1}$ that intersects $\nu_r$. Moreover, no ray in $S\setminus \{r\}$ intersects $\nu_r$ (otherwise $\nu_r$ would not have been present in $V_{i-1}$ by step-5 of MOD-SSR-Algorithm). Hence, all rays in $N(\nu_r)\setminus \{r\}$ must have passed the token to its neighbours. Therefore, for all $x\in N(\nu_r)\setminus \{r\}$ we have $T_{x}=\emptyset$. So, for each ray $r\in S$, there always exists a segment $\nu_r$ such that for all $x\in N(\nu_r)\setminus \{r\}$ we have $T_{x}=\emptyset$. We shall denote such a segment as a \emph{critical segment with respect to $r$} and denote it as $\nu_r$ (in case of multiplicity choose any one as $\nu_r$). Now we have the following lemma.

\begin{observation}\label{obs:constraint-critical} 
	For a ray $r\in S$ let $\nu_r$ be a critical segment with respect to $r$. Then $N(\nu_r)\subseteq T_r$.
\end{observation}
\begin{proof}
	Consider any arbitrary but fixed deleted ray $y\in N(\nu_r)\setminus \{r\}$ which was deleted at some $j^{th}$ iteration. By Observation~\ref{obs:token-pass}, there exists a ray $y'\in R_j$ such that $y'$ intersects $v$ and $y\in T_{y'}$. Now applying the above argument for all rays in $N(\nu_r)\setminus \{r\}$, we have the proof.
\end{proof}

\begin{lemma}\label{lem:2-approx}
	Let $S$ be the set returned by Algorithm~\ref{alg:SSR} with $\mathcal{SSR}(R,V)$ as input and $OPT$ be an optimum solution of $\mathcal{SSR}(R,V)$. Then $|S|\leq 2|OPT|$.
\end{lemma}
\begin{proof}
	
	Let $R$ be the set of rays and $V$ be the set of segments with $|R|=n, |V|=m$. To prove the lemma we consider the following integer linear programming (ILP) formulation $Q$ of $\mathcal{SSR}(R,V)$ and the corresponding relaxed linear programming (LP) formulation $Q_l$ where for each ray $r\in R$, let $x_r\in\{0,1\}$ denote the variable corresponding to $r$.
	
	\begin{center}
		\scalebox{0.85}{
		\begin{tabular}{|c|}
			\hline
			\begin{minipage}{.5\textwidth}
				\begin{equation}
				\begin{array}{ll@{}ll}
				\text{minimize}  & \displaystyle\sum\limits_{r\in R} x_r &\\
				\text{subject to}& \displaystyle\sum\limits_{r\in N(v)} x_r \geq 1,  &\forall v\in V\\
				 x_r \in \{0,1\}, & \forall r\in R
				\end{array}
				\end{equation}
			\end{minipage}\\
			$Q$\\
			\hline
		\end{tabular}}
	\end{center}	
	We shall show that the set $S$ returned by Algorithm~\ref{alg:SSR} gives an integral solution for $Q$ whose cost (i.e. cardinality of $S$) is at most twice the optimum cost of $Q$. 
	This will immediately imply the statement of the lemma. Let $\textbf{Q}_l=\{x_r\}_{r\in R}$ be an optimal solution of $Q_l$. Also define $y_r=1$ if $r\in S$, $y_r=0$ if $r\notin S$ and $\textbf{Q}'=\{y_r\}_{r\in R}$. 
	Now we claim that $\textbf{Q}'$ is a feasible solution of $Q$. This is true because Algorithm~\ref{alg:SSR} terminates only when no segments are left in $V_i$. 
	Hence, for each $v$ there is a ray $r\in S$ that intersects $v$ and therefore each constraint of $Q$ contains a variable $y_r$ such that $y_r=1$ in $\textbf{Q}'$. 
	Hence $\textbf{Q}'$ is feasible. Now we fix any arbitrary $r\in S$ and $\nu_r$ be a critical segment with respect to $r$. Then due to Observation~\ref{obs:constraint-critical}, we know that for all $z\in N(\nu_r)\setminus \{r\}$ we have $T_{z}=\emptyset$ and $N(\nu_r)\subseteq T_r$. 
	Therefore, for the constraint corresponding to $\nu_r$ in $Q_l$, we have that $$\sum_{z\in N(\nu_r)} y_z=1\leq \sum_{z\in N(\nu_r)} x_z \leq \sum_{z\in T_r} x_z \hspace{20pt}\left[\text{since }N(\nu_r)\subseteq T_r\text{ by Observation~\ref{obs:constraint-critical}}\right]$$
	
	Due to Lemma~\ref{lem:at-most-two-copy}, we know that for each ray $r\in R$ there are at most two rays $r_1,r_2$ such that $r\in T_{r_1}\cap T_{r_2}$. Therefore,

	$$|S|=\sum_{r\in S} y_r=\sum_{r\in S} \sum_{z\in N(\nu_r)} y_z\leq \sum_{r\in S} \sum_{z\in T_r} x_z \leq 2\sum_{z\in R} x_z \leq 2|OPT|$$
	
	This completes the proof of the lemma.
\end{proof}

\textbf{Proof of Theorem~\ref{thm:SSR}} follows from Lemma~\ref{lem:2-approx} which essentially implies that SSR-algorithm is a 2-approximation algorithm for \textsc{SSR} problem. We shall use the following corollary in the proof of Theorem~\ref{thm:stab-L-dom} and~\ref{thm:unit-bk-VPG} which follows from Lemma~\ref{lem:2-approx}.

\begin{corollary}\label{cor:SSR-2-approx}
	Let $R$ be a set of leftward-directed rays and $V$ be a set of vertical segments. The cost of an optimal solution for the ILP of $\mathcal{SSR}(R,V)$ is at most $2$ times the cost of an optimal solution for the relaxed LP of $\mathcal{SSR}(R,V)$.
\end{corollary}

\section{Approxmation algorithm for MDS on vertically-stabbed-L graphs}\label{sec:vertical-L}

In this section, we shall give a polynomial time $8$-approximation algorithm to solve the \textsc{MDS} problem on vertically-stabbed-\textsc{L} graphs. In the rest of the paper, $OPT(Q)$ (resp. $OPT(Q_l)$) denotes the cost of the optimum solution of an ILP $Q$ (resp. LP $Q_l$).

\medskip 
\noindent \textbf{Overview of the algorithm:} First, we solve the relaxed LP formulation of the ILP of the \textsc{MDS} problem on the input vertically-stabbed-\textsc{L} graph $G$ and create two subproblems. We shall show that one of those two subproblems is equivalent to the \textsc{SSR} problem and the other is equivalent to a \textsc{Stabbing Rays with Segments} problem (defined below) which was introduced by Katz et al.~\cite{katz2005}. We solve these two subproblems individually and show that the union of the solutions gives a solution for the \textsc{MDS} problem on $G$ which is at most $8$ times the optimal solution.

In the \textbf{\emph{Stabbing Rays with Segments}} (\textsc{SRS}) problem, the input is a set $R$ of (disjoint) leftward-directed rays and a set $V$ of (disjoint) vertical segments. The objective is to select a minimum cardinality subset of $V$ that intersects all rays in $R$. %We shall use the following result.

\begin{result}[\cite{katz2005}]\label{res:KMN}
	There is a $2$-approximation algorithm for the \textit{SRS} problem for $n$ rays and $m$ segments that runs in time $O((m + n) \log(m + n))$, using $O(n + m \log m)$ space.
\end{result}

We shall show that the cost of the optimum solution of the ILP of \textit{SRS} is at most twice the cost of the optimum solution of the corresponding relaxed LP. Below we restate the algorithm of Katz et al.~\cite{katz2005} in a way that would assist our analysis.

\medskip

\noindent\textbf{2-approximation algorithm for SRS problem:} With each segment $v\in V$, we associate a token $T_v$ which is a subset of $V$. Initialise $T_v=\emptyset$ for each $v\in V$. Let $r_i$ be the ray whose right-endpoint, $(x_i , y_i)$, has the smallest $x$-coordinate. (We can assume without loss of generality that $x$- and $y$-coordinates of the endpoints of the rays are all distinct.) Assuming that there is a feasible solution to the \textsc{SRS} instance, there must exist a segment of $V$ that intersects $r_i$. Let $N(r_i)\subseteq V$ be the set of segments that intersect $r_i$. Let $v_{top}$ (resp. $v_{bot}$) be a segment in $N(r_i)$ whose top endpoint is top-most (resp., bottom endpoint is bottom-most); it may be that $v_{top}=v_{bot}$. We add both $v_{top}$ and $v_{bot}$ to our heuristic solution set $S$. Also we set $T_{v_{top}}=T_{v_{bot}}=N(r_i)$. Then we remove from $R$ all of the rays that intersect $v_{top}$ or $v_{bot}$, delete all segments in $N(r_i)$ and then repeat the above steps untill $R=\emptyset$. We shall refer the algorithm stated above as \textsc{KMN}-algorithm.

\medskip

First, we state the following observations required to prove Lemma~\ref{lem:approx-LP-SRS}.

\begin{observation}\label{obs:stab-L-trivial-1}
	For each ray $r$, there is a segment $v\in S$ that intersects $r$.
\end{observation}

\begin{observation}\label{obs:stab-L-trivial-2}
	For each segment $v\in V$, there are at most two tokens such that both of them contains $v$.
\end{observation}

\begin{lemma}\label{lem:approx-LP-SRS}
	Let $R$ (resp. $V$) be a set of (disjoint) leftward-directed rays (resp. vertical segments), $Q$ be the ILP of the \textit{SRS} instance with $R,V$ as input and $Q_l$ be the corresponding relaxed LP. Then $OPT(Q)\leq 2\cdot OPT(Q_l)$.
\end{lemma}

\begin{proof}
	Consider the following ILP of the \textit{SRS} instance with $R,V$ as input. For a ray $u\in R$, let $N(u)$ denote the set of segments in $V$ that intersect $u$.
	\begin{center}
	\scalebox{0.85}{
		\begin{tabular}{|c|}
		\hline
		\begin{minipage}{.65\textwidth}
		\begin{equation}
		\begin{array}{ll@{}ll}
		\text{minimize}  & \displaystyle\sum\limits_{w\in V} x_w &\\
		\text{subject to}& \displaystyle\sum\limits_{w\in N(u)} x_w \geq 1,  &\forall u\in R\\
		x_w\in \{0,1\}, & \forall w\in V
		\end{array}
		\end{equation}
		\end{minipage}\\
		$Q$\\
		\hline
		\end{tabular}}
	\end{center}
	Let $\textbf{X}=\{x_v\}_{v\in V}$ be an optimal solution of $Q_l$ (relaxed LP of $Q$) where $x_v$ denotes the value of the variable in $Q_l$ corresponding to $v\in V$. Let $S$ be the solution returned by the \textsc{KMN} algorithm with $R,V$ as input. Now define for each $v\in V$, $x'_v=1$ if $v\in S$, $x'_v=0$ if $v\notin S$ and let $\textbf{X}'=\{x'_v\}_{v\in V}$. By Observation~\ref{obs:stab-L-trivial-1}, $\textbf{X}'$ is a feasible solution of $Q$. For each $z\in S$, there is a ray $r_i$ such that $T_z=N(r_i)$. Therefore, for the constraint corresponding to $r_i$ in $Q_l$, we have the following
	\begin{center}
		\begin{tabular}{ccc}
			\begin{minipage}{.55\textwidth}
				\begin{equation}\label{ineq:1}
				x'_z= 1 \leq \displaystyle\sum\limits_{v\in N(r_i)} x_v = \displaystyle\sum\limits_{v\in T_{z}} x_v
				\end{equation}
			\end{minipage} 
		\end{tabular}
	\end{center} 
	Hence,
	
	\begin{center}
			
				\begin{minipage}{.55\textwidth}
					\begin{equation} 
					\begin{split}
					|S| & = \displaystyle\sum\limits_{v\in S} x'_v \\
					& \leq \displaystyle\sum\limits_{v\in S} \displaystyle\sum\limits_{v'\in T_v} x_{v'}\hspace{20pt} \text{Using Inequality~\ref{ineq:1}}\\
					& \leq 2\displaystyle\sum\limits_{v'\in V} x_{v'}\hspace{32pt} \text{Using Observation~\ref{obs:stab-L-trivial-2}}\\
					& =  2\cdot OPT(Q_l)
					\end{split}
					\end{equation}
				\end{minipage} 
		\end{center}
	
	This completes the proof.
\end{proof}

Now we are ready to describe our approximation algorithm for \textsc{MDS} problem on vertically-stabbed-\textsc{L} graphs. Let $\mathcal{R}$ be a vertically-stabbed-\textsc{L}-representation of a graph $G$. Without loss of generality, we assume that 
\begin{enumerate}
	
	\renewcommand{\theenumi}{(\roman{enumi})}
	\renewcommand{\labelenumi}{\theenumi}
	
	\item the vertical line $x=0$ intersects all the \textsc{L}-paths in $\mathcal{R}$,
	\item the $x$-coordinate of the corner point of each \textsc{L}-path in $\mathcal{R}$ is strictly less than $0$, and 
	\item whenever two distinct \textsc{L}-paths intersect in $\mathcal{R}$, they intersect at exactly one point.
\end{enumerate}

\begin{figure}
	\centering
	\scalebox{.7}{
		\begin{tabular}{ccccccc}
			\includegraphics[page=12]{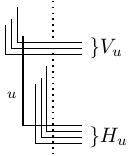}&\hspace{10pt}&\includegraphics[page=13]{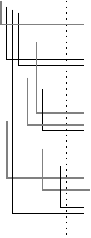}&\hspace{10pt}&\includegraphics[page=14]{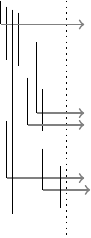}&\hspace{10pt}&\includegraphics[page=15]{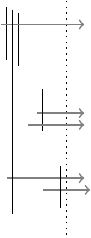}\\
			(a) & \hspace{10pt}& (b) & \hspace{10pt}& (c) & \hspace{10pt}& (d)
		\end{tabular} }
		\caption{A vertically-stabbed-\textsc{L} graph. (a) The sets $H_u$ and $V_u$ corresponding to a vertex $u$. (b) The gray \textsc{L}-paths belongs to $A_1$ and the black \textsc{L}-paths belongs to $A_2$, (c) The subproblem equivalent to \textsc{SRS}-problem. (d) The subproblem equivalent to \textsc{SSR}-problem. }\label{fig:example-Stab-L-Dom}
	\end{figure}
	
	For a vertex $u\in V(G)$, let $N[u]$ denote the closed neghbourhood of $u$ in $G$. For each vertex $u\in V(G)$, let $H_u=\{c\in N[u]\colon \textsc{L}_c~\text{intersects the horizontal segment of}~\textsc{L}_u\}$ and let $V_u$ denote the set $N(u)\setminus H_u$ (See Figure~\ref{fig:example-Stab-L-Dom}(a)). Based on these we have the following ILP (say $Q$) of the problem of finding a minimum dominating set of $G$. 
	\begin{center}
		\scalebox{0.85}{
		\begin{tabular}{|c|}
			\hline
			\begin{minipage}{.65\textwidth}
				\begin{equation}
				\begin{array}{ll@{}ll}
				\text{minimize}  & \displaystyle\sum\limits_{v\in V(G)} x_v &\\
				\text{subject to}& \displaystyle\sum\limits_{v\in H_u} x_v + \displaystyle\sum\limits_{v\in V_u} x_v \geq 1,  &\forall u\in V(G)\\
				x_v\in \{0,1\}, & \forall v\in V(G)
				\end{array}
				\end{equation}
			\end{minipage}\\
			$Q$\\
			\hline
		\end{tabular}}
	\end{center}

	The first step of our algorithm is to solve the relaxed LP formulation, say $Q_l$ of $Q$. Let $\textbf{Q}_l=\{x_v\colon v\in V(G)\}$ be an optimal solution of $Q_l$. Now we define the following sets. $$A_1=\left\{u\in V(G)\colon \displaystyle\sum\limits_{v\in H_u} x_v \geq \frac{1}{2} \right\}, A_2=\left\{u\in V(G)\colon \displaystyle\sum\limits_{v\in V_u} x_v \geq \frac{1}{2} \right\}$$ $$H=\bigcup\limits_{u\in A_1} H_u,V=\bigcup\limits_{u\in A_2} V_u$$ Based on these, we consider the following two integer programs $Q'$ and $Q''$. 
	\begin{center}
		\scalebox{0.85}{
			\begin{tabular}{|c|c|}
				\hline
				\begin{minipage}{.45\textwidth}
					\begin{equation}
					\begin{array}{ll@{}ll}
					\text{minimize}  & \displaystyle\sum\limits_{v\in H} x'_v \\
					\text{subject to}& \displaystyle\sum\limits_{v\in H_u} x'_v \geq 1,  &\forall u\in A_1\\                
					x'_v \in \{0,1\}, &v\in H
					\end{array}            
					\end{equation}
				\end{minipage}&
				\begin{minipage}{.45\textwidth}
					\begin{equation}
					\begin{array}{ll@{}ll}
					\text{minimize}  & \displaystyle\sum\limits_{v\in V} x''_v &\\
					\text{subject to}& \displaystyle\sum\limits_{v\in V_u} x''_v \geq 1,  &\forall u\in A_2\\
					x''_v \in \{0,1\}, &v\in V
					\end{array}
					\end{equation}
				\end{minipage}\\
				$Q'$&$Q''$\\
				\hline
			\end{tabular}}
		\end{center}

		Let $Q'_l$ and $Q''_l$ be the relaxed LP of $Q'$ and $Q''$ respectively. Clearly, the solutions of $Q'$ and $Q''$ gives a solution for $Q$. Hence $OPT(Q)\leq OPT(Q') + OPT(Q'')$. For each $x_v\in \textbf{Q}_l$, define $y_v=\min\{1,2x_v\}$ and define $\textbf{Y}_l=\{y_v\}_{x_v\in \textbf{Q}_l}$. Notice that $\textbf{Y}_l$ gives a solution to $Q'_l$ (and $Q''_l$). Therefore, $OPT(Q'_l) + OPT(Q''_l)\leq 4\cdot OPT(Q_l)$. We have the following lemma.
		
		\begin{lemma}\label{lem:stab-L-dom}
			$OPT(Q')\leq 2\cdot OPT(Q'_l)$ and $OPT(Q'')\leq 2\cdot OPT(Q''_l)$.
		\end{lemma}
		
		\begin{proof}
			Solving $Q'$ is equivalent to finding a minimum cardinality $D\subseteq H$ such that each vertex $u\in A_1$ is adjacent to some vertex in $D\cap H_u$. Recall that for each vertex $u\in A_1$ and $v\in H_u$, $L_v$ intersects the horizontal segment of $L_u$. Now consider the following sets. Let $R$ be the set of horizontal segments of the \textsc{L}-paths representing the vertices in $A_1$ and $S$ be the set of vertical segments of the \textsc{L}-paths representing the vertices in $H$ (See Figures~\ref{fig:example-Stab-L-Dom}(b) and~\ref{fig:example-Stab-L-Dom}(c) for an example).  Since all horizontal segments in $R$ intersect the $x=0$ vertical line and the $x$-coordinates of the vertical segments in $S$ is strictly less than $0$, we can consider the horizontal segments in $R$ as rightward directed rays. Hence, solving $Q'$ is equivalent to solving the ILP (say $\mathcal{E}$) of the problem of finding a minimum cardinality subset of $S$ (a set of vertical segments) that intersects all rays in $R$ (a set of rightward-directed rays). Hence solving $\mathcal{E}$ is equivalent to solving an \textsc{SRS} instance with $R$ and $S$ as input. By Lemma~\ref{lem:approx-LP-SRS}, we have that $$OPT(Q')=OPT(\mathcal{E})\leq 2\cdot OPT(\mathcal{E}_l)\leq 2\cdot OPT(Q'_l)$$ where $\mathcal{E}_l$ is the relaxed LP of $\mathcal{E}$. This proves the first part.
			
			For the second part, solving $Q''$ is equivalent to finding a minimum cardinality subset $D$ of $V$ such that each vertex $u\in A_2$ is adjacent to some vertex in $D\cap V_u$. Recall that, for each vertex $u\in A_2$ and $v\in V_u$, $L_v$ intersects the vertical segment of $L_u$. Now consider the following sets. Let $R$ be the set of horizontal segments of the \textsc{L}-paths representing the vertices in $V$ and $S$ be the set of vertical segments of the \textsc{L}-paths representing the vertices in $A_2$ (See Figures~\ref{fig:example-Stab-L-Dom}(b) and~\ref{fig:example-Stab-L-Dom}(d) for an example). Since all horizontal segments in $R$ intersect the $x=0$ vertical line and the $x$-coordinates of the vertical segments in $S$ is strictly less than $0$, we can consider the segments in $R$ as rightward-directed rays. Hence, solving $Q''$ is equivalent to solving the ILP (say $\mathcal{F}$) of the problem of finding a minimum cardinality subset of $R$ (a set of rightward-directed rays) that intersects all segments in $S$ (a set of vertical segments). Hence, solving $\mathcal{F}$ is equivalent to solving an \textsc{SSR} instance with $R$ and $S$ as inputs. By Corollary~\ref{cor:SSR-2-approx}, we have that $$OPT(Q'')=OPT(\mathcal{F})\leq 2\cdot OPT(\mathcal{F}_l)=2\cdot OPT(Q''_l)$$ where $\mathcal{F}_l$ is the relaxed LP of $\mathcal{F}$. This proves the second part. 
		\end{proof}
		
		\textbf{To complete the proof of Theorem~\ref{thm:stab-L-dom}}, notice that Lemma~\ref{lem:stab-L-dom} implies that solving $Q'$ (resp. $Q''$) is equivalent to solving \textsc{SRS} (resp. \textsc{SSR}) problem instance. Let $A$ be the union of the solutions returned by \textsc{KMN}-algorithm and \textsc{SSR}-algorithm, used to solve $Q'$ and $Q''$ respectively. Hence, $$|A|\leq 2(OPT(Q'_l)+OPT(Q''_l)) \leq 8\cdot OPT(Q_l) \leq 8\cdot OPT(Q)$$ Since the LP $Q_l$ consists of $n$ variables where $n=|V(G)|$, solving $Q_l$ takes $O(n^3\cdot T)$ time~\cite{gonzaga1989} where $T$ is the total number of bits required to encode the $(0,1)$-constraint matrix of $Q_l$. Solving both the SSR and SRS instances takes a total of $O(n\log n)$ time (Theorem~\ref{thm:SSR} and Result~\ref{res:KMN}) and therefore the total running time of the algorithm is $O(n^3\cdot T)$.

		\section{Hardness and Approximation for MDS on Unit-$B_k$ VPG-graphs}\label{sec:unit-VPG}

		\begin{figure}
			\centering
			\scalebox{.6}{
				\begin{tabular}{c@{\hskip 1in}c}
					\includegraphics[page=10]{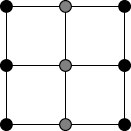}&\includegraphics[scale=0.7,page=11]{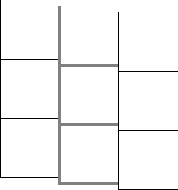}\\
					(a) & (b)
				\end{tabular}}
				\caption{ A unit \textsc{L}-representation (b) of (3,3)-grid (a).}\label{fig:example-grid}
			\end{figure}
			
			First we prove the NP-hardness for \textsc{MDS} problem on unit-$B_1$-VPG graphs. The \emph{$(h,w)$-grid} is the undirected graph $G$ with $V(G)=\{(x,y)\colon x,y\in\mathbb{Z}, 1\leq x\leq h, 1\leq y\leq w\}$ and $E(G)=\{(u,v)(x,y)\colon |u-x|+|v-y|=1\}$. A graph $G$ is a \emph{grid} graph if $G$ is an induced subgraph of $(h',w')$-grid for some positive integers $h',w'$. Given a grid graph $G$, the \textsc{Grid-Dominating-Set} problem is to find a minimum cardinality dominating set of $G$. We shall reduce the \textsc{Grid-Dominating-Set} problem to the \textsc{MDS} problem on unit-$B_1$-VPG graphs. We will be done by showing that for any positive integers $h,w$ the $(h,w)$-grid has a unit-$B_1$-VPG representation. Let $G$ be a $(h,w)$-grid and $\epsilon=\frac{1}{hw}$. For each $(x,y)\in V(G)$ consider the \textsc{L}-path $\textsc{L}_{(x,y)}$ such that top endpoint of vertical segment of $\textsc{L}_{(x,y)}$ is $(x,y-\epsilon(x-1))$ (See Figure~\ref{fig:example-grid}). It is easy to verify that the set of \textsc{L}-paths $\mathcal{R}=\{\textsc{L}_{(x,y)}\colon (x,y)\in V(G)\}$ is a unit-$B_1$-VPG representation of $G$. Hence, we have the following theorem.

			\begin{theorem}
				It is NP-Hard to solve the \textsc{MDS} problem on unit-$B_1$-VPG graphs.
			\end{theorem}

In this section, we describe our approximation algorithm for \textsc{MDS} on Unit-$B_k$ VPG-graphs and prove Theorem~\ref{thm:unit-bk-VPG}. First we give an overview of the algorithm below.

\medskip

\noindent \textbf{Overview of the algorithm for MDS on Unit-$B_k$-VPG graphs:} First, we solve the relaxed LP formulation of the ILP of \textsc{MDS} problem on the input Unit-$B_k$-VPG graph $G$ and create $(k+1)^2$ many subproblems each of which are equivalent to a \textsc{Proper-Seg-Dom} problem (defined below). We shall show that it is possible to get constant factor approximate solutions for each of these subproblems in polynomial time and take the union of all the solutions to get a set $S$. We shall conclude by showing that $S$ is a $O(k^4)$-approximate solution for \textsc{MDS} problem on $G$ and $S$ can be computed in polynomial time.

\medskip
In the \textsc{Proper-Seg-Dom} (\emph{Proper Orthogonal Segment Domination}) problem the input is a set $V$ of vertical segments, a set $H$ of horizontal segments such that the sets of intervals obtained by projecting the segments in $H$ (resp. $V$) onto the $x$-axis (resp. $y$-axis) are proper (a set of intervals is \emph{proper} if no two intervals in the set contain each other). The objective is to select a minimum cardinality subset of $H\cup V$ that intersects all segments in $H\cup V$. 

In Section~\ref{sec:proper-seg-dom} we shall give a constant factor approximation for the \textsc{Proper-Seg-Dom} problem and in Section~\ref{sec:unit-b-k-VPG-approx-complete} we shall describe the approximation algorithm for \textsc{MDS} on Unit-$B_k$ VPG-graphs and use our result on \textsc{Proper-Seg-Dom} to prove Theorem~\ref{thm:unit-bk-VPG}.

\subsection{Approximation for PROPER-SEG-DOM problem}\label{sec:proper-seg-dom}

Let $V$ (resp. $H$) be a set of vertical (resp. horizontal) segments, $\mathcal{I}_V$ (resp. $\mathcal{I}_H$) be the projections of the segments in $V$ (resp. $H$) onto the $y$-axis (resp. $x$-axis) and both $\mathcal{I}_V$ and $\mathcal{I}_H$ are proper sets of intervals. The \textsc{Proper-Seg-Dom} problem is to find a minimum cardinality subset of $H\cup V$ that intersects all segments in $H\cup V$. In this section, we shall give a polynomial time $18$-approximation algorithm to solve the \textsc{Proper-Seg-Dom} problem where $n=|H|+|V|$. Throughout this section, we let $\mathcal{PSD}(H,V)$ denote the \textsc{Proper-Seg-Dom} problem instance with $H$ and $V$ as input. In order to solve \textsc{Proper-Seg-Dom}, we consider the following problems.

\medskip

\noindent \textbf{Overview of the algorithm:} First, we solve the relaxed LP formulation of the ILP of the input \textsc{Proper-Seg-Dom} problem and create two subproblems. We shall show that one of the subproblem is equivalent to a \emph{Subset Proper Interval Domination} problem (defined below) and the other is equivalent to a \emph{Proper Orthogonal Segment Stabbing} problem (defined below). We give constant factor approximate solution for each these two problems and show that the union of the solutions of these subproblems gives a solution for the input \textsc{Proper-Seg-Dom} problem instance which is at most $18$ times the optimal solution.

\medskip

In the \textbf{\emph{Subset Proper Interval Domination}} problem (\textsc{SPID}) the inputs are a proper set of intervals $S$ and a set $T\subseteq S$. The \textsc{SPID} problem is to select a minimum cardinality subset of $S$ that intersects all intervals in $T$. In the \textbf{\emph{Proper Orthogonal Segment Stabbing}} problem (\textsc{POSS}) inputs are a set $S$ of horizontal segments a set $T$ of vertical segments such that the set of intervals obtained by projecting the segments in $S$ onto the $x$-axis is a proper set of intervals. The \textsc{POSS} problem is to select a minimum cardinality subset of $S$ that intersects all vertical segments in $T$. For a set of horizontal segments $S$ and a set of vertical segments $T$, $\mathcal{P}(T,S)$ shall denote the \textsc{POSS} instance. Let $S$ be a proper set of intervals and $T$ be a subset of $S$. Then $\mathcal{SP}(T,S)$ shall denote the \textsc{SPID} instance.

\medskip

In Section~\ref{subsubsec:SPID-POSC}, we give approximation algorithms to solve \textsc{SPID}-problems and \textsc{POSS}-problems. In Section~\ref{subsubsec:Proper-Seg-Dom}, we give the approximation algorithm for the $\mathcal{PSD}(H,V)$ problem.

\subsubsection{Approximation algorithms for SPID and POSS problem}\label{subsubsec:SPID-POSC}

Lemma~\ref{lem:PID-integrality-gap} and Lemma~\ref{lem:approx-POSC} shall compare the integral cost with the fractional cost of the \textsc{SPID}-problem and \textsc{POSS}-problem respectively. We shall use Corollary~\ref{cor:hor-vertical-SPID} and Corollary~\ref{cor:hor-vertical-POSC} in Section~\ref{subsubsec:Proper-Seg-Dom}.

\begin{lemma}\label{lem:PID-integrality-gap}
	Let $S$ be a proper set of intervals and $T\subseteq S$. The cost of an optimal solution for the relaxed LP of $\mathcal{SP}(T,S)$ equals to the cost of an optimal solution for the ILP of $\mathcal{SP}(T,S)$. 
\end{lemma}

\begin{proof}
	For an interval $v\in S$, let $l(v)$ and $ r(v)$ denote the left and right endpoints. Let $s_1,s_2,\ldots,s_k$ be the intervals in $S$ sorted in the ascending order of the right endpoints. Hence, $r(s_1)< r(s_2)< \ldots < r(s_k)$ and as no two intervals in $S$ contain each other, we have $l(s_1) < l(s_2) <\ldots l(s_k)$. For an interval $x\in T$, let $N(x)$ denote the set of intervals in $S$ that intersect $x$. Now consider the following ILP (say $Q$) of $\mathcal{SP}(T,S)$.
	
	\begin{center}
		\begin{tabular}{|c|}
			\hline
			\begin{minipage}{.65\textwidth}
				\begin{equation}
				\begin{array}{ll@{}ll}
				\text{minimize}  & \displaystyle\sum\limits_{v\in S} x_v &\\
				\text{subject to}& \displaystyle\sum\limits_{v\in N(u)} x_v \geq 1,  &\forall u\in T\\
				x_v\in \{0,1\}, & \forall v\in S
				\end{array}
				\end{equation}
			\end{minipage}\\
			$Q$\\
			\hline
		\end{tabular}
	\end{center}

	Let $\mathcal{M}$ be a coefficient matrix of $Q$ such that the $i^{th}$ column of $\mathcal{M}$ corresponds to the variable corresponding to $s_i\in S$. Each row of $\mathcal{M}$ have the interval property (i.e the set of 1's are consecutively in each row). To see this consider any three intervals $\{s_i,s_j,s_{j'}\}\subseteq S$ such that $i<j<j'$ and any interval $t\in T$ such that $t$ intersects both $s_i$ and $s_{j'}$. Since $r(s_i) < r(s_j) < r(s_{j'})$ and $l(s_i)<l(s_j) < l(s_{j'})$, $t$ must intersect $s_{j}$. Therefore, $\mathcal{M}$ is a totally unimodular matrix~\cite{schrijver1998}. Thus any optimal solution of the relaxed LP of $\mathcal{SP}(T,S)$ is integral and therefore an optimal solution for the ILP of $\mathcal{SP}(T,S)$. This completes the proof.
\end{proof}

For a proper set of intervals $S$ and $T\subseteq S$, an optimal solution of $\mathcal{SP}(T,S)$ can be computed in $O(n\log n)$ time~\cite{chang1998}. We shall use the following corollary whose proof follows from Lemma~\ref{lem:PID-integrality-gap}.

\begin{corollary}\label{cor:hor-vertical-SPID}
	Let $X$ (resp. $Y$) be a set of horizontal (resp. vertical) segments, $\mathcal{I}_X$ (resp. $\mathcal{I}_Y$) be the projections of the segments in $X$ (resp. $Y$) onto the $x$-axis (resp. $y$-axis) and both $\mathcal{I}_X$ and $\mathcal{I}_Y$ are proper sets of intervals. For $X'\subseteq X$ and $Y'\subseteq Y$, let $\mathcal{A}$ be the ILP of finding a minimum cardinality subset $D$ of $X\cup Y$ such that every horizontal (resp. vertical) segment in $X'$ (resp. $Y'$) intersects at least one horizontal (resp. vertical) segment in $D\cap X$ (resp. $D\cap Y$). Then $OPT(\mathcal{A})=OPT(\mathcal{A}_l)$ where $\mathcal{A}_l$ is the relaxed LP of $\mathcal{A}$. Moreover, $OPT(\mathcal{A})$ can be computed in $O(n\log n)$ time where $n=|X|+|Y|$.
\end{corollary}

\begin{lemma}\label{lem:approx-POSC}
	Let $S$ be a set of horizontal segments, $T$ be a set of vertical segments and $\mathcal{I}_S$ be the projections of the horizontal segments in $S$ onto the $x$-axis. Let $\mathcal{I}_S$ be a proper set of intervals. Then the cost of an optimal solution for the ILP of $\mathcal{P}(T,S)$ is at most $8$ times the cost of an optimal solution for the relaxed LP of $\mathcal{P}(T,S)$. 
\end{lemma}

\begin{proof}
	Without loss of generality, we assume that the segments in $S$ and $T$ are in ``general position" i.e. $x$-coordinate of any vertical segment in $T$ is distinct from the $x$-coordinates of the left and right endpoints of any interval in $\mathcal{I}_S$. Since no two interval in $\mathcal{I}_S$ contain each other, we have a set $P$ of real numbers such that each interval in $\mathcal{I}_S$ contains exactly one real number from $P$. (To see this, consider the $x$-coordinates of the right endpoints of the intervals in the maximum cardinality subset of $\mathcal{I}_S$ with pairwise non-intersecting intervals which is obtained using the greedy algorithm~\cite{kleinberg2006}). Add in $P$ two more dummy values $q,q'$ which are not contained in any interval in $\mathcal{I}_S$ and $q$ (resp. $q'$) is less than (resp. greater than) that of all values in $P$. Let $p_1,p_2,\ldots,p_t$ be the values in $P$ sorted in the ascending order (notice that $p_1=q$ and $p_t=q'$). For each $i\in\{1,2,\ldots,t-1\}$, let $T_i$ denote the vertical segments of $T$ that lies inside the strip bounded by the lines $y=p_i$ and $y=p_{i+1}$. Due to our general position assumption for any $i\neq j$, $T_i$ and $T_j$ are disjoint. For each $i\in \{1,2,\ldots,t-1\}$, and each vertical segment $v\in T_i$, let $S_v^{left}$ (resp. $S_v^{right}$) be the subset of $S$ that intersects $v$ and the line $y=p_i$ (resp. $y=p_{i+1}$). Since any interval in $\mathcal{I}_S$ contains exactly one value from $P$ and therefore from $\{p_i,p_{i+1}\}$, $S_v^{left}\cap S_v^{right}=\emptyset$, for each vertical segment $v\in T$. Based on these we have the following ILP (say $W$) of $\mathcal{P}(T,S)$ problem.
	
	\begin{center}
		\begin{tabular}{|c|}
			\hline
			\begin{minipage}{.65\textwidth}
				\begin{equation}
				\begin{array}{ll@{}ll}
				\text{minimize}  & \displaystyle\sum\limits_{v\in S} x_v &\\
				\text{subject to}& \displaystyle\sum\limits_{v\in S^{left}_u} x_v + \displaystyle\sum\limits_{v\in S^{right}_u} x_v \geq 1,  &\forall u\in T\\
				x_v\in \{0,1\}, & \forall v\in S
				\end{array}
				\end{equation}
			\end{minipage}\\
			$W$\\
			\hline
		\end{tabular}
	\end{center}
	
	First step of our algorithm is to solve the relaxed LP formulation (say $W_l$) of $W$. Let $\textbf{W}_l=\{x_v\colon v\in S\}$ be an optimal solution of $W_l$. Consider the following sets. $$A_1=\left\{u\in T\colon \displaystyle\sum\limits_{v\in S^{left}_u} x_v \geq \frac{1}{2} \right\}, A_2=\left\{u\in T\colon \displaystyle\sum\limits_{v\in S^{right}_u} x_v \geq \frac{1}{2} \right\}$$ $$L=\bigcup\limits_{v\in A_1} S^{left}_v, R=\bigcup\limits_{v\in A_2} S^{right}_v$$ Based on these, we consider the following two integer programs $W'$ and $W''$.

	\begin{center}
		\begin{tabular}{|c|c|}
			\hline
			\begin{minipage}{.45\textwidth}
				\begin{equation}
				\begin{array}{ll@{}ll}
				\text{minimize}  & \displaystyle\sum\limits_{v\in L} x'_v \\
				\text{subject to}& \displaystyle\sum\limits_{v\in S^{left}_u} x'_v \geq 1,  &\forall u\in A_1\\                
				x'_v \in \{0,1\}, &v\in L
				\end{array}            
				\end{equation}
			\end{minipage}&
			\begin{minipage}{.45\textwidth}
				\begin{equation}
				\begin{array}{ll@{}ll}
				\text{minimize}  & \displaystyle\sum\limits_{v\in R} x''_v \\
				\text{subject to}& \displaystyle\sum\limits_{v\in S^{right}_u} x''_v \geq 1,  &\forall u\in A_2\\                
				x''_v \in \{0,1\}, &v\in R
				\end{array}
				\end{equation}
			\end{minipage}\\
			$W'$&$W''$\\
			\hline
		\end{tabular}
	\end{center}

	Let $W'_l$ and $W''_l$ be the corresponding relaxed LPs of $W'$ and $W''$ respectively. The union of the solutions of $W'$ and $W''$ gives a solution for $W$ implying $OPT(W)\leq OPT(W') + OPT(W'')$. For each $x_v\in \textbf{W}_l$, define $y_v=\min\{1,2x_v\}$ and define $\textbf{Y}_l=\{y_v\}_{x_v\in \textbf{W}_l}$. Notice that $\textbf{Y}_l$  gives a solution to $W'_l$ (and $W''_l$). Hence, $OPT(W'_l)\leq 2\cdot OPT(W_l)$ and $OPT(W''_l)\leq 2\cdot OPT(W_l)$. Therefore, $OPT(W'_l) + OPT(W''_l)\leq 4\cdot OPT(W_l)$. Notice that, solving $W'$ (resp. $W''$) is equivalent to the problem of finding a minimum cardinality subset of the horizontal segments in $L$ (resp. $R$) to intersect all vertical segments in $A_1$ (resp. $A_2$). Now we have the following claim.
	
	\begin{claim}
		$OPT(W')\leq 2\cdot OPT(W'_l)$ and $OPT(W'')\leq 2\cdot OPT(W''_l)$.
	\end{claim}
	
	\medskip
	\noindent We shall prove the above claim only for $W'$ as proof for the other case is similar. Recall that solving $W'$ is equivalent to the problem of finding a minimum cardinality subset of the horizontal segments in the set $L$ (defined earlier) to intersect all vertical segments in $A_1$. For each $i\in\{1,2,\ldots,(t-1)\}$ let $T_{1,i}=A_1\cap T_i$ and $L_i$ be the set of horizontal segments in $L$ that intersect some vertical segment in $T_{1,i}$. Formally, $L_i=\bigcup\limits_{v\in T_{1,i}} S^{left}_v$. For any $i\neq j$, $T_{1,i}\cap T_{1,j}=\emptyset$ and $L_i\cap L_j=\emptyset$ (this follows from the fact no horizontal segment in $S$ intersects both $y=p_i$ and $y=p_j$). For each $i\in \{1,2,\ldots,(t-1)\}$, let $\mathcal{D}_i$ (resp, $\mathcal{D}_{i,l}$) denote the ILP (resp. relaxed LP) of the problem of selecting minimum subset $D_i$ horizontal segments in $L_i$ such that all vertical segments in $T_{1,i}$ intersect at least one horizontal segment in $D_i$. Clearly, $OPT(W')=\sum_{i=1}^{t-1} OPT(\mathcal{D}_i)$ and $OPT(W'_l)=\sum_{i=1}^{t-1} OPT(\mathcal{D}_{i,l})$. For each $i\in \{1,2,\ldots,(t-1)\}$ notice that, all horizontal segments intersect the vertical line $y=p_i$ and all vertical segments in $T_{1,i}$ lies to the left of the vertical line $y=p_i$. For each $i\in \{1,2,\ldots,(t-1)\}$ if we consider the segments in $L_i$ to be leftward-directed rays then solving $\mathcal{D}_i$ is equivalent to solving an \textsc{SSR} instance with $T_{1,i}$ and $L_i$ as input. Due to Corollary~\ref{cor:SSR-2-approx}, for each $i\in\{1,2,\ldots,(t-1)\}$, $OPT(\mathcal{D}_i) \leq 2\cdot OPT(\mathcal{D}_{i,l})$. Hence, $$OPT(W')=\displaystyle\sum\limits_{i=1}^{t-1} OPT(\mathcal{D}_i)\leq 2\cdot \displaystyle\sum\limits_{i=1}^{t-1} OPT(\mathcal{D}_{i,l})=2\cdot OPT(W'_l)$$ This completes the proof of the claim.
	
	\medskip
	\noindent
	Using the above claim and previous observations, we can infer that $$OPT(W) \leq OPT(W') + OPT(W'') \leq 2(OPT(W'_l)+OPT(W''_l)) \leq 8\cdot OPT(W'_l)$$ This completes the proof of the lemma.
\end{proof}

Therefore, we have a polynomial time $8$-approximation algorithm to solve the $\mathcal{P}(T,S)$ problem where $n=|S|+|T|$. We shall use the following corollary which follows from Lemma~\ref{lem:approx-POSC}.

\begin{corollary}\label{cor:hor-vertical-POSC}
	Let $X_1,X_2$ (resp. $Y_1,Y_2$) be sets of horizontal (resp. vertical) segments, $\mathcal{I}_X$ (resp. $\mathcal{I}_Y$) be the projections of the segments of $X_1\cup X_2$ (resp. $Y_1\cup Y_2$) onto the $x$-axis (resp. $y$-axis) and both $\mathcal{I}_X$ and $\mathcal{I}_Y$ are proper sets of intervals. Let $\mathcal{B}$ be the ILP of finding minimum cardinality subset $D$ of $X_2\cup Y_2$ such that every horizontal (resp. vertical) segment in $X_1$ (resp. $Y_1$) intersect at least one vertical (resp. horizontal) segment in $D\cap Y_2$ (resp. $D\cap X_2$). Then there is a polynomial time algorithm to compute a set $D'$ which gives a solution of $\mathcal{B}$ and $|D'|\leq 8\cdot OPT(\mathcal{B}_l)$ where $n=|X_1\cup X_2\cup Y_1\cup Y_2|$ and $\mathcal{B}_l$ is the relaxed LP of $\mathcal{B}$.
\end{corollary}

\subsubsection{Approximation algorithms for $\mathcal{PSD}(H,V)$}\label{subsubsec:Proper-Seg-Dom}

Now we are ready to describe our approximation algorithm to solve the $\mathcal{PSD}(H,V)$ problem. Recall that $V$ (resp. $H$) is a set of vertical (resp. horizontal) segments, $\mathcal{I}_V$ (resp. $\mathcal{I}_H$) is the projections of the segments in $V$ (resp. $H$) onto the $y$-axis (resp. $x$-axis) and both $\mathcal{I}_V$ and $\mathcal{I}_H$ are proper sets of intervals. The $\mathcal{PSD}(H,V)$ problem is to find a minimum cardinality subset of $H\cup V$ that intersects all vertical segments in $H\cup V$. For a segment $v\in H\cup V$, let $N(v)\subseteq H\cup V$ denote the set of segments that intersects $v$. For a segment $w\in H$, let $N_o(w)=N(w)\cap H$ and for a segment $w'\in V$ let $N_o(w')=N(w')\cap V$. Based on these we have the following ILP (say $Z$) for the $\mathcal{PSD}(H,V)$ problem.

\begin{center}
	\begin{tabular}{|c|}
		\hline
		\begin{minipage}{.65\textwidth}
			\begin{equation}
			\begin{array}{ll@{}ll}
			\text{minimize}  & \displaystyle\sum\limits_{w\in H\cup V} x_w &\\
			\text{subject to}& \displaystyle\sum\limits_{w\in N_o(u)} x_w + \displaystyle\sum\limits_{w\in N(u)\setminus N_o(u)} x_w \geq 1,  &\forall u\in H\cup V\\
			x_w\in \{0,1\}, & \forall w\in H\cup V
			\end{array}
			\end{equation}
		\end{minipage}\\
		$Z$\\
		\hline
	\end{tabular}
\end{center}

\sloppy The first step of our algorithm is to solve the relaxed LP formulation (say $Z_l$) of $Z$. Let $\textbf{Z}_l=\{x_w\colon w\in H\cup V\}$ be an optimal solution of $Z_l$. Let $$A_1=\left\{u\in H\cup V\colon \displaystyle\sum\limits_{w\in N_o(u)} x_w \geq \frac{1}{2} \right\}, A_2=\left\{u\in H\cup V\colon \displaystyle\sum\limits_{w\in N(u)\setminus N_o(u)} x_w \geq \frac{1}{2} \right\}$$, $$B_1=\bigcup\limits_{u\in A_1} N_o(u), \hspace{5pt} B_2=\bigcup\limits_{u\in A_2} N(u)\setminus N_o(u)$$ Based on these, we consider the following two integer programs $Z'$ and $Z''$.

\begin{center}
	\begin{tabular}{|c|c|}
		\hline
		\begin{minipage}{.45\textwidth}
			\begin{equation}
			\begin{array}{ll@{}ll}
			\text{minimize}  & \displaystyle\sum\limits_{w\in B_1} x'_w \\
			\text{subject to}& \displaystyle\sum\limits_{w\in N_o(v)} x'_w \geq 1,  &\forall v\in A_1\\                
			x'_w \in \{0,1\}, &w\in B_1
			\end{array}            
			\end{equation}
		\end{minipage}&
		\begin{minipage}{0.45\textwidth}
			\begin{equation}
			\begin{array}{ll@{}ll}
			\text{minimize}  & \displaystyle\sum\limits_{w\in B_2} x''_w \\
			\text{subject to}& \displaystyle\sum\limits_{w\in N(v)\setminus N_o(v)} x''_w \geq 1,  &\forall v\in A_2\\                
			x''_w \in \{0,1\}, &w\in B_2
			\end{array}
			\end{equation}
		\end{minipage}\\
		$Z'$&$Z''$\\
		\hline
	\end{tabular}
\end{center}

Let $Z'_l$ and $Z''_l$ be the corresponding relaxed LPs of $Z'$ and $Z''$ respectively. Clearly, the union of the solutions of $Z'$ and $Z''$ gives a solution for $Z$. Hence, $OPT(Z)\leq OPT(Z') + OPT(Z'')$. For each $x_v\in \textbf{Z}_l$, define $y_v=\min\{1,2x_v\}$ and define $\textbf{Y}_l=\{y_v\}_{x_v\in \textbf{Z}_l}$. Notice that $\textbf{Y}_l$ gives a solution for $Z'_l$ and $Z''_l$. Hence, $OPT(Z'_l) \leq 2\cdot OPT(Z_l)$ and $OPT(Z''_l)\leq 2\cdot OPT(Z_l)$. Now we prove the following lemma. 

\begin{lemma}\label{lem:PSD-1}
	$OPT(Z')=OPT(Z'_l)$ and $OPT(Z'')\leq 8\cdot OPT(Z''_l)$.
\end{lemma}

\begin{proof}
	
	To prove the first part, let $X$ (resp. $Y$) be the set of horizontal (resp. vertical) segments in $B_1$ and $X'$ (resp. $Y'$) be the set of horizontal (resp. vertical) segments in $A_1$. Notice that $X'\subseteq X$ and $Y'\subseteq Y$. Hence, $Z'$ is the ILP of finding minimum cardinality subset $D$ of $X\cup Y$ such that every horizontal (resp. vertical) segment in $X'$ (resp. $Y'$) intersects at least one horizontal (resp. vertical) segment in $D\cap X$ (resp. $D\cap Y$). By Corollary~\ref{cor:hor-vertical-SPID}, we have that $OPT(Z')=OPT(Z'_l)$.
	
	To prove the second part, let $X_1$ and $X_2$ (resp. $Y_1$ and $Y_2$) be the sets of horizontal (resp. vertical) segments in $A_2$ and $B_2$, respectively. Notice that $Z''$ is the ILP of finding minimum cardinality subset $D$ of $X_2\cup Y_2$ such that every horizontal (resp. vertical) segment in $X_1$ (resp. $Y_1$) intersects at least one vertical (resp. horizontal) segment in $D\cap Y_2$ (resp. $D\cap X_2$). By Corollary~\ref{cor:hor-vertical-POSC}, we have that $OPT(Z'')\leq 8\cdot OPT(Z''_l)$.
\end{proof}

Using Lemma~\ref{lem:PSD-1}, we can conclude that $$OPT(Z) \leq OPT(Z') + OPT(Z'') \leq OPT(Z'_l) + 8\cdot OPT(Z''_l) \leq 18\cdot OPT(Z'_l)$$ 

Due to Lemma~\ref{lem:PSD-1}, Corollary~\ref{cor:hor-vertical-SPID} and~\ref{cor:hor-vertical-POSC} we have the following theorem.

\begin{theorem}\label{thm:Proper-Seg-Dom-2}
	There is a polynomial time $18$-approximation algorithm to solve the \textsc{Proper-Seg-Dom} problem where $n$ is the total number of segments in the input instance. 
\end{theorem}

We shall use the following corollary in Section~\ref{sec:unit-b-k-VPG-approx-complete}, whose proof follows from that of Theorem~\ref{thm:Proper-Seg-Dom-2}.

\begin{corollary}\label{cor:proper-seg-dom}
	Let $V_1,V_2$ (resp. $H_1,H_2$) be a set of vertical (resp. horizontal) segments, $\mathcal{I}_V$ (resp. $\mathcal{I}_H$) be the projections of the segments of $V_1\cup V_2$ (resp. $H_1\cup H_2$) onto the $y$-axis (resp. $x$-axis) and both $\mathcal{I}_V$ and $\mathcal{I}_H$ are proper sets of intervals. Let $\mathcal{C}$ be the ILP of the problem of finding a minimum cardinality subset $D$ of $V_2\cup H_2$ such that every segment in $V_1\cup H_1$ intersects some segment in $D$. Then there is a polynomial time algorithm to compute a set $D'$ which gives a solution of $\mathcal{C}$ and $|D'|\leq 18\cdot OPT(\mathcal{C}_l)$ where $n=|V_1\cup V_2\cup H_1\cup H_2|$ and $\mathcal{C}_l$ is the relaxed LP of $\mathcal{C}$.
\end{corollary}

\subsection{Completion of Proof for Theorem~\ref{thm:unit-bk-VPG}}\label{sec:unit-b-k-VPG-approx-complete}

Let $\mathcal{R}$ be a unit $B_k$-VPG representation of a unit $B_k$-VPG graph $G$. We shall assume that every vertex of $G$ has a self loop (this does not contradict the intersection model as every rectilinear path intersects itself). For a vertex $v\in V(G)$, let $P(v)$ denote the path in $\mathcal{R}$ that corresponds to $v$. For a vertex $v\in V(G)$, $N(v)$ and $N[v]$ denote the \emph{open neighbourhood} and \emph{closed neighbourhood} of $v$, respectively. Throughout this section, we shall assume that the segments of each path $P\in \mathcal{R}$ are numbered consecutively starting from the leftmost segment by $1,2,\ldots,t$ where $t (\leq k+1)$ is the number of segments in $P$. 

\begin{definition}
	Let $\phi\colon E(G)\rightarrow \mathbb{N}\times \mathbb{N}$ be a mapping such that for an edge $uv$, $\phi(uv)=(i,j)$ if and only if
	\begin{enumerate}
		\item the $i^{th}$ segment of $P(u)$ intersects the $j^{th}$ segment of $P(v)$, and
		
		\item for all $1\leq a<i$ and $1\leq b<j$, the $a^{th}$ segment of $P(u)$ and $b^{th}$ segment of $P(v)$ does not intersect each other.
	\end{enumerate}
\end{definition}

Notice that for each vertex $u\in V(G)$, $\phi(u,u)=(1,1)$. Based on the definition we can partition the closed neighbourhood of a vertex as follows. For a vertex $u\in V(G)$, let $X_u(i,j)=\{v\in N(u) \colon \phi(uv)=(i,j)\}$. For distinct pairs $(i,j)$ and $(i',j')$ the sets $X_u(i,j)$ and $X_u(i',j')$ are disjoint. By $\mathcal{K}$ we shall denote the set $\{1,2,\ldots,k+1\}\times \{1,2,\ldots,k+1\}$. Based on these we have the following ILP of the \textsc{Unit-$B_k$-VPG-Dom} problem on $G$.

\begin{center}
	\begin{tabular}{|c|}
		\hline
		\begin{minipage}{.65\textwidth}
			\begin{equation}
			\begin{array}{ll@{}ll}
			\text{minimize}  & \displaystyle\sum\limits_{v\in V(G)} x_v &\\
			&\\
			\text{subject to}& \displaystyle\sum\limits_{(i,j)\in \mathcal{K}}\hspace{5pt}\sum\limits_{v\in X_u(i,j)} x_v \geq 1,  & \hspace{5pt}\forall u\in V(G)\\
			&\\
			x_v\in \{0,1\}, & \forall v\in V(G)
			\end{array}
			\end{equation}
		\end{minipage}\\
		$Q$\\
		\hline
	\end{tabular}
\end{center}

First step of our algorithm is to solve the relaxed LP formulation (say $Q_l$) of $Q$. Let $\mathbf{Q}_l=\{x_v\colon v\in V(G) \}$ be an optimal solution of $Q_l$. For each vertex $u\in V(G)$, there is a pair $(i,j)\in \mathcal{K}$ such that $\sum\limits_{v\in X_u(i,j)} x_v \geq \frac{1}{(k+1)^2}$. For each pair $(i,j)\in \mathcal{K}$, define $$A(i,j)=\left\{u\in V(G)\colon \sum\limits_{v\in X_u(i,j)} x_v \geq \frac{1}{(k+1)^2}\right\}$$

$$B(i,j)=\displaystyle\bigcup\limits_{u\in A(i,j)} X_{u}(i,j)$$

Based on these we have the following ILP for each pair $(i,j)\in \mathcal{K}$.

\begin{center}
	\begin{tabular}{|c|}
		\hline
		\begin{minipage}{.65\textwidth}
			\begin{equation}
			\begin{array}{ll@{}ll}
			\text{minimize}  & \displaystyle\sum\limits_{v\in B(i,j)} x'_v &\\
			&\\
			\text{subject to}& \displaystyle\sum\limits_{v\in X_u(i,j)} x'_v\geq 1,  & \hspace{5pt}\forall u\in A(i,j)\\
			&\\
			x'_v\in \{0,1\}, & \forall v\in B(i,j)
			\end{array}
			\end{equation}
		\end{minipage}\\
		\\$Q(i,j)$\\
		\hline
	\end{tabular}
\end{center}

For each pair pair $(i,j)\in \mathcal{K}$, let $Q_l(i,j)$ be the relaxed LP of $Q(i,j)$. We have the following $$OPT(Q)\leq \displaystyle\sum\limits_{(i,j)\in \mathcal{K}} OPT(Q(i,j))$$ 

For each $x_v\in \textbf{Q}_l$, define $y_v=\min\{1,2x_v\}$ and define $\textbf{Y}_l=\{y_v\}_{x_v\in \textbf{Q}_l}$. Clearly, $\textbf{Y}_l$ gives a solution to $Q_l(i,j)$ for each $(i,j)\in \mathcal{K}$. Hence, we have the following inequality. $$\displaystyle\sum\limits_{(i,j)\in \mathcal{K}} OPT(Q_l(i,j)) \leq (k+1)^4\cdot OPT(Q_l)$$

Now we have the following lemma.

\begin{lemma}\label{lem:unit-B-VPG}
	For each pair $(i,j)\in \mathcal{K}$, there is a solution $D(i,j)$ for $Q(i,j)$ such that $|D(i,j)|\leq 18\cdot OPT(Q_l(i,j))$.
\end{lemma}

\begin{proof}
	Fix an arbitrary but fixed pair $(i,j)\in \mathcal{K}$. Solving $Q(i,j)$ is equivalent to finding a minimum cardinality subset $D$ of $B(i,j)$ such that each vertex $u\in A(i,j)$ has a neighbour in $D\cap X_u(i,j)$. Since for each vertex $u\in A(i,j)$ and each $v\in X_u(i,j)$, the $i^{th}$ segment of $P(u)$ intersects the $j^{th}$ segment of $P(v)$. Now we define some sets as follows.
	$$S=\{i^{th}\text{ segment of } P(u)\colon u\in A(i,j)\}, \hspace{5pt}T=\{j^{th}\text{ segment of } P(v)\colon v\in B(i,j)\}$$
	$$V_1=\text{set of vertical segments in }S , \hspace{5pt}V_2=\text{set of vertical segments in }T $$
	$$H_1=\text{set of horizontal segments in }S ,\hspace{5pt} H_2=\text{set of horizontal segments in }T $$
	
	Notice that, solving $Q(i,j)$ is equivalent to the problem finding a minimum cardinality subset $D$ of $V_2\cup H_2$ such that every segment in $V_1\cup H_1$ intersect at least one segment in $D$. Since every segment in $V_1\cup V_2\cup H_1\cup H_2$ have unit length, we can modify the vertical (resp. horizontal) segments in $V_1,V_2$ (resp. $H_1, H_2$) and get $V'_1,V'_2$ (resp. $H'_1\cup H'_2$) respectively such that 
	
	\begin{enumerate}
		\item the set of intervals obtained by projecting the segments in $V'_1\cup V'_2$ (resp. $H'_1\cup H'_2$) onto the $y$-axis (resp. $x$-axis) is a proper set of intervals, and
		
		\item two segments $u',v'\in V'_1\cup V'_2\cup H'_1\cup H'_2$ intersect if and only if the corresponding segments $u,v\in V_1\cup V_2\cup H_1\cup H_2$ intersect.
	\end{enumerate}
	
	Hence, solving $Q(i,j)$ is equivalent to the problem of finding a minimum cardinality subset $D$ of $V'_2\cup H'_2$ such that every segment in $V'_1\cup H'_1$ intersect at least one segment in $D$. Moreover, the sets of intervals obtained by projecting the segments in $V'_1\cup V'_2$ and $H'_1\cup H'_2$ onto the $y$-axis and $x$-axis, respectively are proper sets of intervals. Hence by Corollary~\ref{cor:proper-seg-dom}, we can find a solution (say $D(i,j)$) for $Q(i,j)$ such that $|D(i,j)|\leq 18\cdot OPT(Q_l(i,j))$. This completes the proof.
\end{proof}

For each pair $(i,j)\in \mathcal{K}$, due to Lemma~\ref{lem:unit-B-VPG}, we can get a solution $D(i,j)$ of $Q(i,j)$ such that $|D(i,j)|\leq 18\cdot OPT(Q(i,j))$ in polynomial time. Let $D$ be the union of $D(i,j)$'s for all $(i,j)\in \mathcal{K}$. We have that $$|D|=\displaystyle\sum\limits_{(i,j)\in \mathcal{K}} |D(i,j)|\leq \displaystyle\sum\limits_{(i,j)\in \mathcal{K}} 18\cdot OPT(Q_l(i,j))\leq 18\cdot (k+1)^4\cdot OPT(Q_l)\leq 18\cdot (k+1)^4\cdot OPT(Q) $$

This completes the Proof of Theorem~\ref{thm:unit-bk-VPG}.

			\section{Conclusion}\label{sec:conclude}
			
			In this paper, we gave approximation algorithms for \textsc{MDS} problems on vertically-stabbed-\textsc{L} graphs and unit-$B_k$-VPG graphs. We gave a polynomial time $O(k^4)$-approximation algorithm to solve the \textsc{MDS} problem on unit-$B_k$-VPG graphs. However, the status of \textsc{MDS} problems on $B_k$-VPG graphs remains unknown for a fixed $k$. Approximation algorithms with better running time or better approximation factor for \textsc{MDS} problems on vertically-stabbed-\textsc{L} graphs and unit-$B_k$-VPG graphs would be interesting.

\bibliography{reference}

\end{document}